%% file: paper.tex
\pdfoutput=1

\documentclass[11pt]{article}

\usepackage[numbers,sort&compress]{natbib}
\usepackage{amsthm}
\usepackage{amssymb}
\usepackage{hyperref}
\usepackage[figure,algo2e,linesnumbered,vlined]{algorithm2e}
\usepackage{algorithm}
\usepackage[noend]{algpseudocode}
\usepackage{graphicx}
\usepackage{xspace}
\usepackage{url}
\usepackage{subcaption}
\usepackage{appendix}
\usepackage{authblk}

\graphicspath{{./pictures/}}

\usepackage[T1]{fontenc}
\usepackage[tt=false]{libertine}
\usepackage[libertine]{newtxmath}
\usepackage{inconsolata}
\usepackage{geometry}
\SetAlFnt{\scriptsize}

\usepackage{mathtools}

\newtheorem{theorem}{Theorem}

\usepackage{booktabs}

\title{A Scalable, Portable, and Memory-Efficient Lock-Free FIFO Queue}

\author{Ruslan Nikolaev}
\affil{rnikola@vt.edu\\Virginia Tech, USA}

\date{}

\begin{document}

\maketitle

\begin{abstract}
\input{abstract.tex}
\end{abstract}

\input{commands.tex}


\input{introduction.tex}

\input{background.tex}

\input{ds.tex}

\input{ring.tex}

\input{ringadv.tex}

\input{correctness.tex}

\input{evaluation.tex}

\input{related.tex}

\input{conclusion.tex}

\bibliography{lockfree}

\end{document}

%% file: abstract.tex
We present a new lock-free multiple-producer and multiple-consumer
(MPMC) FIFO queue design which is scalable and, unlike existing high-performant
queues, very memory efficient. Moreover, the design is ABA safe and does not
require any external memory allocators or safe memory reclamation techniques,
typically needed by other scalable designs. In fact, this queue itself can be
leveraged for object allocation and reclamation, as in data pools.
We use FAA (fetch-and-add), a specialized and more scalable than CAS
(compare-and-set) instruction, on the most contended hot spots of the
algorithm. However, unlike prior attempts with FAA, our queue is both
lock-free and linearizable.

We propose a general approach, SCQ, for bounded queues. This approach can
easily be extended to support unbounded FIFO queues which can store
an arbitrary number of elements. SCQ is portable across virtually all
existing architectures
and flexible enough for a wide variety of uses. We measure the performance
of our algorithm on the x86-64 and PowerPC architectures. Our evaluation
validates that our queue has exceptional memory efficiency compared to other
algorithms and its performance is often comparable to, or
exceeding that of state-of-the-art scalable algorithms.

%% file: commands.tex
\algnewcommand{\Null}{\textbf{nullptr}}%
\algnewcommand{\algorithmicgoto}{\textbf{goto}}%
\algnewcommand{\Goto}[1]{\algorithmicgoto~\ref{#1}}%
\algdef{SE}[DOWHILE]{Do}{doWhile}{\algorithmicdo}[1]{\algorithmicwhile\ #1}%
\algnewcommand\Not{\textbf{!}}
\algnewcommand\AndOp{\textbf{and}\xspace}
\algnewcommand\ModOp{\textbf{mod}\xspace}
\algnewcommand\OrOp{\textbf{or}\xspace}

\SetKwIF{If}{ElseIf}{Else}{if (}{)}{else if}{else}{endif}

\SetKwRepeat{Do}{do}{while}
\SetKwProg{Fn}{}{}{}
\newcommand{\removelatexerror}{\let\@latex@error\@gobble}

\def\keywords{\vspace{.5em}
{\noindent{\textit{Keywords}:\,\relax%
}}}
\def\endkeywords{\par}

\algnewcommand{\LineComment}[1]{\State \(\triangleright\) #1}

%% file: introduction.tex
\section{Introduction}
\label{sec:intro}

Creating efficient concurrent algorithms for modern multi-core architectures
is challenging. Efficient and scalable lock-free FIFO
queues proved to be especially hard to devise.
Although elimination techniques address LIFO
stack performance~\cite{Hendler:2004:SLS:1007912.1007944},
their FIFO counterparts~\cite{Moir:2005:UEI:1073970.1074013} are somewhat
more restricted:
a \textit{dequeue} operation can eliminate an \textit{enqueue}
operation only if all preceding queue entries have already been consumed.
Thus, FIFO elimination is more suitable in specialized cases and, typically,
shorter queues.
Relaxed versions of FIFO queues, which can reorder elements, were also proposed~\cite{Kirsch:2013:FSL:2960356.2960376}, but they cannot be used
when a strict FIFO ordering is required. FIFO
queues are important
in many applications, especially in fixed-size data pools which
use bounded \textit{circular queues} (\textit{ring buffers}).

Most correct linearizable and lock-free implementations of such queues rely
on the compare-and-set (CAS) instruction. Although this instruction
is powerful, it does not scale well as the contention grows.
However, less powerful instructions such as fetch-and-add (FAA)
are known to scale better.
In Figure~\ref{fig:faa}, we demonstrate execution time for
FAA vs. CAS measured in a tight loop for the corresponding
number of threads. The results are for
an almost ``ideal case'' to simply emulate FAA in a CAS loop.
As CAS loops in typical algorithms are more complex,
the actual gap is bigger.

FAA may seem to be a natural fit for ring buffers when updating queue's
head and tail pointers. The reality, however, is more
nuanced when designing lock-free queues.
Many straight-forward algorithms without explicit locks that use FAA, e.g.,~\cite{Krizhanovsky:2013:LMM:2492102.2492106}, are actually not lock-free~\cite{Feldman:2015:WMM:2835260.2835264} because it is possible to find
an execution pattern where no thread can make progress.
Lack of true lock-freedom manifests in suboptimal performance when
some threads are preempted since other threads are effectively blocked when the preemption happens
in the middle of a queue operation. Additionally, such queues
cannot be safely used in environments where blocking is not permitted.
Even if FAA is not used, certain queues~\cite{vyakov} fail to achieve
linearizable lock-free behavior.
A case in point: an open source lock-free
data structure library,
liblfds~\cite{liblfds}, simply falls back~\cite{ringdisapp} to
the widely known Michael \& Scott's (M\&S) FIFO lock-free
queue~\cite{Michael:1998:NAP:292022.292026} in their ring buffer
implementation. While easy to implement, this queue does not
scale well as we show in Section~\ref{sec:eval}.

Despite the aforementioned challenges, the use of FAA is re-invigorated by recent concurrent FIFO
queue designs~\cite{Yang:2016:WQF:2851141.2851168,Morrison:2013:FCQ:2442516.2442527}. Unfortunately,~\cite{Yang:2016:WQF:2851141.2851168,Morrison:2013:FCQ:2442516.2442527}, despite their good performance,
are not always memory efficient as we demonstrate in Section~\ref{sec:eval}.
Furthermore, such queues rely on memory allocators
and safe memory reclamation schemes, thus creating a ``chicken and egg''
situation when allocating memory blocks. For example, if we simply want to
recycle memory blocks or create a queue-based memory pool for allocation,
reliance on an external memory allocator to allocate and deallocate
memory is undesirable. Furthermore, typical
system memory allocators, including jemalloc~\cite{jemalloc}, are not
lock-free. Lock-based allocators can
defeat the purpose of creating a purely lock-free algorithm since
they weaken overall progress guarantees. Moreover, in a number of use cases,
such as within OS kernels, blocking can be prohibited or undesirable.

\begin{figure}
\begin{minipage}{.5\textwidth}
\includegraphics[width=\textwidth]{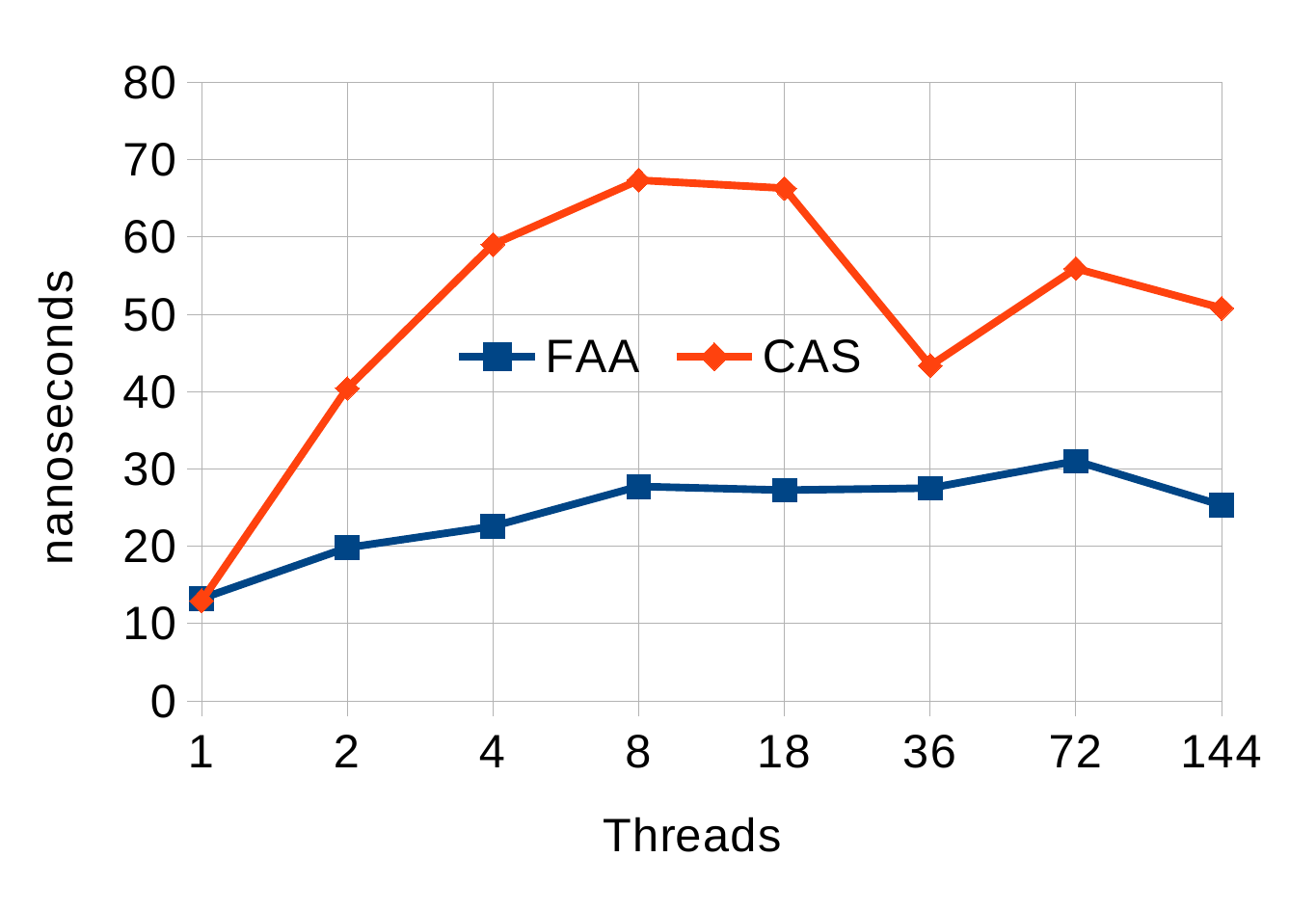}
\caption{FAA vs. CAS on 4x18 Xeon E7-8880.}
\label{fig:faa}
\end{minipage}%
\begin{minipage}{.5\textwidth}
\begin{algorithm2e}[H]
\textbf{int} Tail = 0, Head = 0\tcp*{Queue's tail and head}
\textbf{void *} Array[$\infty$]\tcp*{An infinite array}
\Fn{\textbf{void} enqueue(\textbf{void *} p)} {
\While {\upshape True} {
	T = FAA(\&Tail, 1)\;
	\tcp{Repeat the loop if the entry is}
	\tcp{already invalidated by dequeue()}
	\If {\upshape SWAP(\&Array[T], p) = $\bot$} {
		\Return\;
	}
}
}
\Fn{\textbf{void *} dequeue()} {
\While {\upshape True} {
	H = FAA(\&Head, 1)\;
	p = SWAP(\&Array[H], $\top$)\;
	\lIf {\upshape p $\neq$ $\bot$} {\Return p}
	\If {\upshape Load(\&Tail) $\le$ H + 1} {
	\Return \Null\tcp*{Empty}
}
}
}
\caption{Infinite array queue (susceptible to livelocks).}
\label{alg:infring1}
\end{algorithm2e}
\end{minipage}
\end{figure}

\paragraph*{Contributions of the paper}
\begin{itemize}
\item We introduce an approach of building ring buffers using indirection
and two queues.
\item We present our \textit{scalable circular queue} (SCQ) design which uses FAA.
To the
best of our knowledge, it is the first ABA-safe design that is scalable,
livelock-free and relies only on single-width atomic operations.
It is inspired by CRQ~\cite{Morrison:2013:FCQ:2442516.2442527} but
prevents livelocks and uses our indirection approach.
Although CRQ attempts to solve a similar problem, it
is livelock-prone, uses double-width CAS (unavailable on PowerPC~\cite{ppc:manual}, MIPS~\cite{mips:manual}, RISC-V~\cite{riscv:manual},
SPARC~\cite{sparc:manual}, and
other architectures) and is not standalone; it uses M\&S
queue as an extra
layer (LCRQ) to work around livelock situations.
\item Since unbounded queues are also used widely, we present the LSCQ design (Section~\ref{sec:unbounded}) which chains SCQ ring buffers in a list. LSCQ is more memory efficient than LCRQ.
\end{itemize}

%% file: background.tex
\section{Background}
\label{sec:background}

\paragraph*{Lock-free algorithms}
We consider an algorithm lock-free if at least one thread
can make progress
in a finite number of steps. In other words, individual threads may
starve, but a preempted thread must not block other threads
from making progress. In contrast, spin locks (either implicit,
found in some incorrect algorithms, or explicit) will prevent other threads
from making further
progress if the thread holding the lock is scheduled out.

\paragraph*{Atomic primitives}
CAS (compare-and-set) is used universally by most lock-free algorithms.
However, one downside of CAS is that it can fail, especially
under large contention. Although specialized instructions
such as FAA (fetch-and-add) and SWAP do not reduce memory contention
directly, they are more efficiently implemented by hardware and
never fail. FAA and SWAP are currently implemented by x86-64~\cite{intel:manual},
ARMv8.1+~\cite{arm:manual}, and RISC-V~\cite{riscv:manual}.

\paragraph*{Safe memory reclamation (SMR)}
Most non-trivial lock-free algorithms, including queues,
require special treatment of memory blocks that need
to be deallocated, as concurrent threads may still access memory referred to by
pointers retrieved prior to the change in a corresponding lock-free data
structure.
For programming languages such as C/C++, where unmanaged code is prevalent,
lock-free \textit{safe memory reclamation} techniques such
as hazard pointers~\cite{hazardPointers} are used. The main high-level idea is that
each accessed pointer must be protected by a corresponding
API call. When done, the thread's pointer reservation can be reset.
When SMR knows that a memory block can be returned safely to the OS, it
triggers memory deallocation.
Bounded SCQ does not need SMR, but certain other lock-free queues such
as LCRQ rely on SMR by their design.

\paragraph*{Infinite array queue}
Figure~\ref{alg:infring1} shows an infinite array queue, originally
described for the LCRQ design~\cite{Morrison:2013:FCQ:2442516.2442527}.
This queue is susceptible to livelocks, but our infinite array queue as
well as the SCQ design are inspired by it.
Initially, the queue is empty, i.e., all its entries are set
to a special $\bot$ value.
\textit{enqueue} uses FAA on \texttt{Tail} to retrieve a slot which will be used
to place a new entry. An enqueuer will try to use this slot. However,
if the previous value is not $\bot$, some dequeuer already modified it,
and this enqueuer moves on to the next slot.
\textit{dequeue} uses FAA on \texttt{Head} to retrieve a slot which contains
a previously produced entry. A dequeuer will insert another special value,
$\top$, to indicate
that the slot can no longer be used. If the previous value is not $\bot$,
a corresponding enqueuer has already produced some entry, which is taken by this
dequeuer. Otherwise, this dequeuer moves on to the next slot. A corresponding
enqueuer, which arrives later, will be unable to use this slot.

%% file: ds.tex
\section{Preliminaries}
\label{sec:prelim}

\paragraph*{Assumptions}
We assume that a program has $k$ threads that can run on
any number of physical CPU cores. For the purpose of this work,
we will assume that the maximum (bounded) queue size is $n$. As
no thread should block on another thread, we will further
reasonably assume that $k\le n$.

For simplicity of our presentation, we will assume that
the system memory model is sequentially
consistent~\cite{Lamport:1979:MMC:1311099.1311750}. However,
the actual implementations of the algorithms (including implementations used in Section~\ref{sec:eval}) can rely on weaker memory models whenever possible.

\paragraph*{Data structure}
Our design is based on two key ideas. First, we use
indirection, i.e., data entries are not stored in the queue itself. Instead,
a queue entry simply records an index into the array of data. Second,
we maintain two queues: {\bf fq}, which keeps indices to
unallocated entries of the array, and {\bf aq}, which keeps allocated
indices to be consumed. A producer thread dequeues an index
from {\bf fq}, writes data to the corresponding array entry, and inserts
the index into {\bf aq}. A consumer thread dequeues the index from {\bf aq},
reads data from the array, and inserts the entry back into {\bf fq}.

Both queues maintain {\tt Head} and {\tt Tail} references (Figure~\ref{fig:ringbuffer}).
They are incremented when new entries are enqueued ({\tt Tail})
or dequeued ({\tt Head}). At any point, these references can be represented
as $j+i\times n$, where $j$ is an \textit{index} (position in the ring buffer),
$i$ is a {\it cycle},
and $n$ is a ring buffer size (must be power of 2 in our implementation).
For example, for {\tt Head}, we can calculate
index $j=(Head\mod n)$ and cycle $i=(Head\div n)$.

Queue entries mirror {\tt Head} and {\tt Tail} values.
Each entry also records an index into
the array, pointing to the data associated with the entry. Unlike {\tt Head}
and {\tt Tail}, it suffices to just record {\it cycle}~$i$,
as entry positions are redundant. We instead record an index into the
array that is of the same bit-length as the position.

\paragraph*{ABA safety}
The ABA problem is prevented by comparing cycles. As both
{\tt Head} and {\tt Tail}
are incremented sequentially, regardless of queue size, they will not wrap
around until after the number of operations exceeds the CPU word's largest
value, a reasonable assumption made by other
ABA-safe designs as well.

\begin{figure}
\begin{minipage}{.5\textwidth}
\includegraphics[width=.95\textwidth]{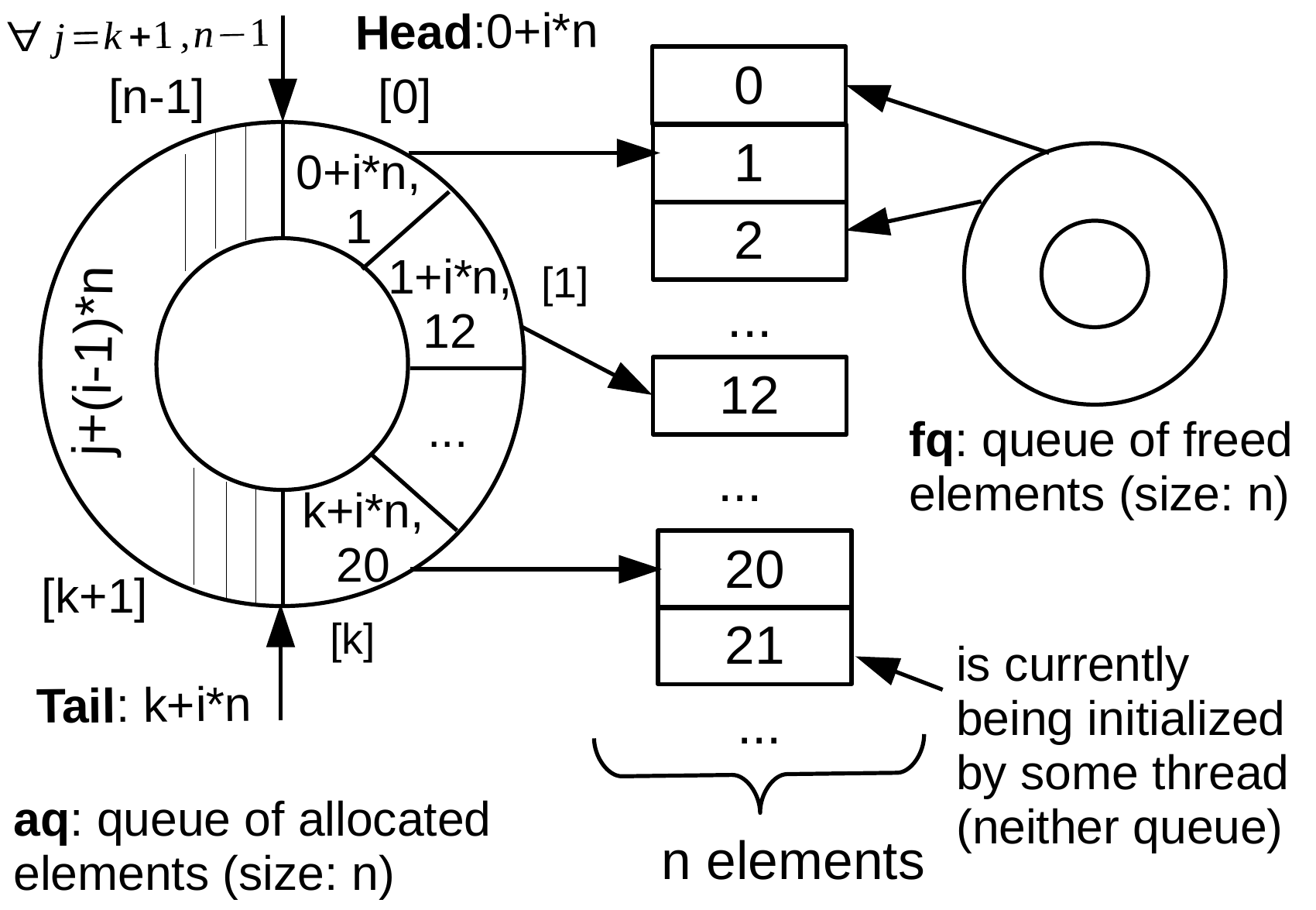}
\caption{Proposed data structure.}
\label{fig:ringbuffer}
\end{minipage}%
\begin{minipage}{.5\textwidth}
\begin{algorithm2e}[H]
\tcp{\textbf{data}: an array of pointers}
\tcp{\textbf{aq} is initialized empty}
\tcp{\textbf{fq} is initialized full}
\BlankLine
\Fn{\textbf{bool} enqueue\_ptr(\textbf{void *} ptr)} {
\textbf{int} index = \textbf{fq}.dequeue()\;
\lIf (\tcp*[f]{Full}) {\upshape index = $\varnothing$} {\Return False}
data[index] = ptr\;
\textbf{aq}.enqueue(index)\;
\Return True\tcp*{Success}
}
\BlankLine
\Fn{\textbf{void *} dequeue\_ptr()} {
\textbf{int} index = \textbf{aq}.dequeue()\;
\lIf (\tcp*[f]{Empty}) {\upshape index = $\varnothing$} {\Return \textbf{nullptr}}
ptr = data[index]\;
\textbf{fq}.enqueue(index)\;
\Return ptr\tcp*{Success}
}
\BlankLine
\caption{Example: storing pointers.}
\label{alg:ringptr}
\end{algorithm2e}
\end{minipage}%
\end{figure}

\paragraph*{Data entries and pointers}
Data array entries can be of any type and size.
It is not uncommon for programs to use a pair of queues with recyclable
elements, e.g.,
queues analogous to \textbf{aq} and \textbf{fq}.
In this case, a program can simply use indices
instead of pointers.
It is also
possible to simply store (arbitrary) pointers as data entries.
We can build a FIFO queue with data pointers
using {\bf aq} and {\bf fq} queues as shown in Figure~\ref{alg:ringptr}.
A producer thread dequeues an entry
from {\bf fq}, initializes it with a pointer and inserts
the entry into {\bf aq}. A consumer thread dequeues the entry from {\bf aq},
reads the pointer and inserts the entry back into {\bf fq}.
Note that \textit{enqueue} does not need to check if a queue is full.
It is only called when an available entry (out of $n$) exists
(e.g., can be dequeued from \textbf{fq} if enqueueing to \textbf{aq}, or vice versa).

%% file: ring.tex
\section{Naive Circular Queue (NCQ)}

\label{sec:ring}

\begin{figure}
\begin{subfigure}{.5\textwidth}
\begin{algorithm2e}[H]
\textbf{int} Tail = n, Head = n\tcp*{Initialization}
\ForAll{\upshape \textbf{entry\_t} Ent $\in$ Entries[n]} {Ent = \{ .Cycle: 0, .Index: 0 \}\;}
\Fn{\textbf{void} enqueue(\textbf{int} index)} {
\Do {\upshape \Not CAS(\&Entries[j], Ent, New)}{
T = Load(\&Tail)\; \label{eloop}
j = Cache\_Remap(T \ModOp n)\;
Ent = Load(\&Entries[j])\;
\If {\upshape Cycle(Ent) = Cycle(T)} {
	CAS(\&Tail, T, T + 1)\tcp*{Help to}

	\Goto{eloop}\tcp*{move tail}

}
\If {\upshape Cycle(Ent) + 1 $\neq$ Cycle(T)} {
	\Goto{eloop}\tcp*{T is already stale}
}
New = \{ Cycle(T), index \}\;
}
CAS(\&Tail, T, T+1)\tcp*{Try to move tail}
}
\end{algorithm2e}
\end{subfigure}%
\begin{subfigure}{.5\textwidth}
\begin{algorithm2e}[H]
\setcounter{AlgoLine}{16}
\Fn{\textbf{int} dequeue()} {
\Do {\upshape \Not CAS(\&Head, H, H+1)}{
H = Load(\&Head)\; \label{dlop}
j = Cache\_Remap(H \ModOp n)\;
Ent = Load(\&Entries[j])\;
\If {\upshape Cycle(Ent) $\neq$ Cycle(H)} {
	\If {\upshape Cycle(Ent) + 1 $=$ Cycle(H)} {\Return $\varnothing$\tcp*{Empty queue}}
	\Goto{dlop}\tcp*{H is already stale}

}
}
\Return Index(Ent)\;
}
\end{algorithm2e}
\end{subfigure}%
\caption{Naive circular queue (NCQ).}
\label{alg:ring}
\end{figure}

Let us first consider a simple algorithm, NCQ, which uses the presented
data structure but borrows an idea of moving queue's tail on behalf of another thread from M\&S queue~\cite{Michael:1998:NAP:292022.292026}. It achieves performance similar to M\&S queue but
does not need double-width CAS to avoid the ABA problem.
We use this queue as an extra baseline in Section~\ref{sec:eval}.

Figure~\ref{alg:ring} shows the {\it enqueue} and {\it dequeue} operations.
In the algorithm, we assume ordinary unsigned integer ring arithmetic
when calculating cycles.
Empty queues initialize all entries to cycle~0.
Their {\tt Head} and {\tt Tail} are both~$n$
(cycle~1). Full queues initialize all entries to cycle~0
along with allocated entry indices. Their {\tt Head} is~$0$ (cycle~0)
and {\tt Tail} is~$n$ (cycle~1).

Entries are always
updated sequentially. To reduce contention due to false sharing, we remap
queue entry positions by using a simple permutation function, {\it Cache\_Remap},
that places two adjacent entries
into different cache lines. The function remaps entries such that the same
cache line will not be reused in the ring buffer as long as possible.

When dequeuing,
we verify that {\tt Head}'s cycle number matches the cycle number of the entry {\tt Head}
is pointing to. If {\tt Head} is one cycle ahead, the queue is empty.
Any other mismatches (i.e., an entry is ahead)
imply that a producer has recycled this entry already. Therefore,
other threads must have already consumed the entry and incremented {\tt Head}
since then (i.e., the previously loaded {\tt Head} value is stale).

As discussed in Section~\ref{sec:prelim}, enqueueing is only possible when an available entry exists.
To successfully enqueue an entry, {\tt Tail} must
be one cycle ahead of an entry it currently points to. {\tt Tail}'s
cycle number equals the entry's cycle number if another thread has already inserted
an element but has not yet advanced {\tt Tail}. The current thread
helps to advance {\tt Tail} to facilitate global progress.
All other cycle mismatches (i.e., an entry is ahead) imply
that the fetched {\tt Tail} value is stale, as there has been at least an
entire round ($n$~enqueues) since it was last loaded.

%% file: ringadv.tex
\section{Scalable Circular Queue (SCQ)}

\label{sec:ringadv}
We will now consider a more elaborated design.
Our scalable circular queue (SCQ) is partially inspired by
CRQ~\cite{Morrison:2013:FCQ:2442516.2442527} as well as
by our data structure (Section~\ref{sec:prelim}). The major problem
with CRQ
is that it is not standalone due to its inherent susceptibility
to livelocks. CRQ must be coupled with a truly lock-free FIFO
queue (such as M\&S queue~\cite{Michael:1998:NAP:292022.292026}). If a livelock happens while
enqueueing entries, a slow path is taken, where the current CRQ instance is
``closed''. Then a new CRQ instance is
allocated and used to enqueue new entries. This approach replaces CRQ with a
list of CRQs (LCRQ). As discussed in introduction, this
design has to rely on a memory allocator and memory reclamation
scheme.

SCQ is not only standalone and
livelock-free, but also much
more portable across different CPU architectures. Unlike CRQ that requires
a special double-width CAS instruction, our SCQ algorithm only
needs single-width CAS,
available across virtually all modern architectures.
SCQ enables support for PowerPC~\cite{ppc:manual} where CRQ/LCRQ cannot be implemented~\cite{Yang:2016:WQF:2851141.2851168,Morrison:2013:FCQ:2442516.2442527}.
Similarly, SCQ enables support for MIPS~\cite{mips:manual}, SPARC~\cite{sparc:manual}, and RISC-V~\cite{riscv:manual} which, like PowerPC,
do not support double-width CAS.

\subsection{Infinite array queue}

\label{sec:infin}
We start off with the presentation of our infinite array queue. We
diverge from the original idea (Section~\ref{sec:background}) in
two major ways.

First, we present a solution to livelocks caused by dequeuers
by introducing a special ``threshold'' value that we describe below.
Livelocks occur when dequeuers incessantly invalidate slots that enqueuers
are about to use for their new entries. By using the threshold value,
we do not carry over the livelock problem to the practical implementation
as in case of CRQ, i.e., guarantee that at least one enqueuer as well as one dequeuer both succeed
after a finite number of steps at any point of time. Algorithms with this property were
previously called
\textit{operation-wise} lock-free~\cite{Morrison:2013:FCQ:2442516.2442527}. This term represents a stronger version of lock-freedom.

Second, our queue reflects the
design presented in Section~\ref{sec:prelim},
where we use indices into the array of data rather than
pointers. This approach guarantees that \textit{enqueue} always
succeeds (neither \textbf{aq} nor \textbf{fq} ever ends up with more
than $n$ elements). Rather, both ``full'' and ``empty'' conditions are detected
by the corresponding \textit{dequeue} operation as in Figure~\ref{alg:ringptr}.
Consequently, \textit{enqueue} does not need to be treated specially to detect
full queues.

\begin{figure}
\begin{subfigure}{.5\textwidth}
\begin{algorithm2e}[H]
\tcp{Threshold prevents livelocks}
\textbf{signed int} Threshold = -1\tcp*{Empty queue}
\BlankLine
\Fn{\textbf{void} enqueue(\textbf{int} index)} {
\While {\upshape True} {
	T = FAA(\&Tail, 1)\;
	\If {\upshape SWAP(\&Entries[T], index) = $\bot$} {
		Store(\&Threshold, $2n - 1$)\;
		\Return\;
	}
}
}
\end{algorithm2e}
\end{subfigure}%
\begin{subfigure}{.5\textwidth}
\begin{algorithm2e}[H]
\setcounter{AlgoLine}{7}
\Fn{\textbf{int} dequeue()} {
\lIf {\upshape Load(\&Threshold) $<$ 0} {\Return $\varnothing$}
\While {\upshape True} {
	H = FAA(\&Head, 1)\;
	index = SWAP(\&Entries[H], $\top$)\;
	\lIf {\upshape index $\neq$ $\bot$} {\Return index}
	\If {\upshape FAA(\&Threshold, -1) $\le$ 0} {\Return $\varnothing$}
	\lIf {\upshape Load(\&Tail) $\le$ H + 1} {\Return $\varnothing$}
}
}
\end{algorithm2e}
\end{subfigure}%
\caption{Infinite array queue. (We make it livelock-free by using a ``threshold''.)}
\label{alg:infring2}
\end{figure}

We note that a circular queue accommodates only a finite number
of elements, $n$. Furthermore, per assumptions in Section~\ref{sec:prelim},
the number of concurrent enqueuers or dequeuers never exceeds $n$ ($k\le n$).
We apply these restrictions to our infinite
queue and present a modified algorithm in Figure~\ref{alg:infring2}.

Suppose that \textit{enqueue} successfully inserts some entry
into the queue and sets the threshold on Line~6 prior to completion.
Let us consider the very last inserted entry for which the threshold
is set. (The threshold can also be set by any concurrent
\textit{enqueue} for preceding entries if their Line~6 executes
after last enqueuer's Line~5.) We will justify the threshold value later.
If dequeuers are active at that moment, we have either of the two scenarios:

\textit{\textbf{The last dequeuer is not ahead of the inserted entry.}}
In this case, the
last dequeuer is no farther than $n$ slots to the left of
the inserted entry (Figure~\ref{fig:livelock}). This is because we have at most
$n$ available slots which can be referenced by any concurrent enqueuers.
None of the enqueuers in this region (i.e., after the last dequeuer) can
fail because the corresponding slots are not being invalidated by the dequeuers.
(Note that this argument is only applicable to the infinite array queue. We will
further refine it for SCQ below.)

\textit{\textbf{The last dequeuer gets ahead of the inserted entry (either initially, or
in the process of dequeuing).}}
Any preceding concurrent dequeuer
either succeeds when its entry can be consumed (Line~13), or fails and retries.
Since \texttt{Head} increases monotonically (Line~11), if a failed dequeuer
ever retries, its new position can only be after the current position
of the last
dequeuer. However, since the inserted entry is the very last
for which \textit{enqueue} is complete (Line~6), all of the dequeuers located after
the last dequeuer (inclusively) are doomed to fail unless some other
concurrent enqueuer completes. (In the latter case, we recursively go back
to the very beginning of our argument where we have chosen the last inserted entry.)

\begin{figure}[ht!]
\centering
\includegraphics[width=.5\columnwidth]{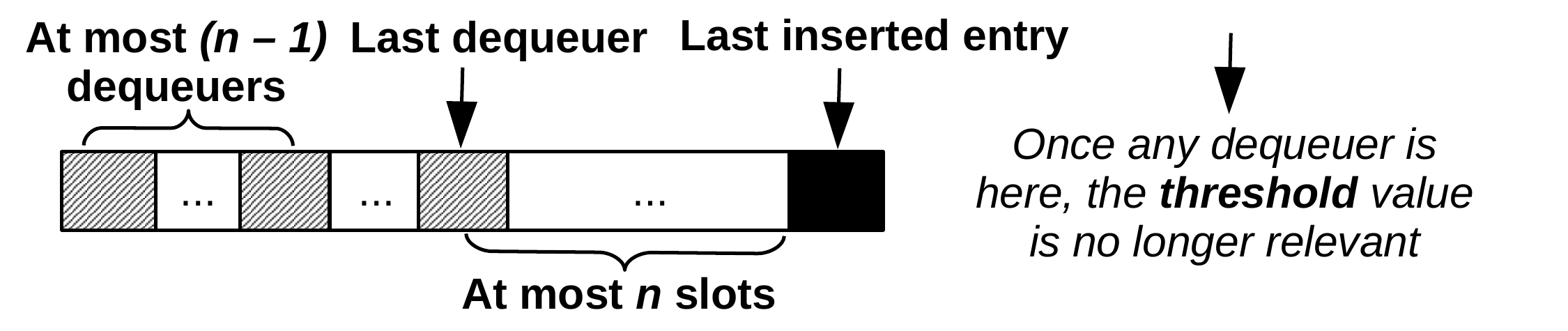}
\caption{Threshold bound for livelock prevention.}
\label{fig:livelock}
\end{figure}

In the second scenario, we are not concerned about the threshold value.
In other words, Lines~14-15 that terminate dequeuers after the inserted entry
do not cause any problems. All these dequeuers are guaranteed to fail
one way or the other. Some preceding dequeuer, which is already in
progress, will eventually fetch the inserted entry. However, in the first
scenario, we want to make sure that Line~14 does not prematurely terminate
dequeuers that are still trying to reach the inserted entry. Specifically, any
failed attempt is penalized by decreasing the threshold value (Line~14).

To reach the last inserted entry, the last dequeuer, or any dequeuer that
follows it later, will unsuccessfully traverse at most $n$ slots
(Figure~\ref{fig:livelock}). At the moment when the entry
is inserted, we may also have up to $n-1$ dequeuers (since $k\le n$) that
are lagging
behind (any distance away). When they fail, they are penalized by subtracting
the threshold value. When they retry, the are guaranteed to be
after the last dequeuer (Line~11). Thus, to guarantee that a dequeuer
eventually reaches the inserted entry, the threshold must be $2n-1$.

Note that the threshold approach carefully avoids memory contention on
the fast path in \textit{dequeue}. Only \textit{enqueue} typically
updates this value through an ordinary memory write with a barrier.
Moreover, if the threshold is still intact, it does not need to
be updated.

\subsection{SCQ algorithm}
\label{sec:scqalg}
In Figure~\ref{alg:ringadv}, we present our SCQ algorithm.
It is based on the modified infinite array queue
and (partially) CRQ~\cite{Morrison:2013:FCQ:2442516.2442527}.
It manages
cycles differently in \textit{dequeue}, making it possible to leverage
a simpler atomic OR operation instead of CAS.

Every entry in the SCQ buffer consists of the \textit{Cycle} and \textit{Index} components.
We also reserve one bit in each entry, \textit{IsSafe}, that
we describe below. This bit is similar to the corresponding bit in CRQ.

\begin{figure}
\begin{subfigure}{.5\textwidth}
\begin{algorithm2e}[H]
\textbf{int} Tail = 2n, Head = 2n\tcp*{Empty queue}
\textbf{signed int} Threshold = -1\;
\ForAll{\upshape \textbf{entry\_t} Ent $\in$ Entries[2n]} {Ent = \{ .Cycle=0, .IsSafe=1, .Index=$\bot$ \}\;}
\Fn (\tcp*[f]{Internal}) {\textbf{void} catchup(\textbf{int} tail, \textbf{int} head)} {
\While {\upshape \Not CAS(\&Tail, tail, head)} {
	head = Load(\&Head)\;
	tail = Load(\&Tail)\;
	\If {tail $\ge$ head} {
		\textbf{break}\;
	}
} 
}
\Fn{\textbf{void} enqueue(\textbf{int} index)} {
\While {\upshape True} {
T = FAA(\&Tail, 1)\;
j = Cache\_Remap(T \ModOp 2n)\;
Ent = Load(\&Entries[j])\;\label{ealoop}
	\If {\upshape Cycle(Ent) < Cycle(T) $\AndOp$ Index(Ent) = $\bot$ $\AndOp$ \hspace{5em} (IsSafe(Ent) \OrOp Load(\&Head) $\le$ T)} {
New = \{ Cycle(T), 1, index \}\;
\If {\upshape \Not CAS(\&Entries[j], Ent, New)} {\Goto{ealoop}}
	\If {\upshape Load(\&Threshold) $\ne 3n-1$} {Store(\&Threshold, $3n - 1$)}
\Return\;
}
}
}
\end{algorithm2e}
\end{subfigure}%
\begin{subfigure}{.5\textwidth}
\begin{algorithm2e}[H]
\setcounter{AlgoLine}{22}
\Fn{\textbf{int} dequeue()} {
	\If (\tcp*[f]{Check if}) {\upshape Load(\&Threshold) $<$ 0} {
		\Return $\varnothing$\tcp*{the queue is empty}
	}
\While {\upshape True} {
H = FAA(\&Head, 1)\;
j = Cache\_Remap(H \ModOp 2n)\;
Ent = Load(\&Entries[j])\;\label{daloop}
\If {\upshape Cycle(Ent) = Cycle(H)} {
	\tcp{Cycle can't change, mark as $\bot$}
	Atomic\_OR(\&Entries[j], \{ 0, 0, $\bot$ \})\;
	\Return Index(Ent)\tcp*{Done}
}
New = \{ Cycle(Ent), 0, Index(Ent) \}\;
\uIf {\upshape Index(Ent) = $\bot$} {
	New = \{ Cycle(H), IsSafe(Ent), $\bot$\}\;
}
\If {\upshape Cycle(Ent) < Cycle(H)} {
\If {\upshape \Not CAS(\&Entries[j], Ent, New)} {\Goto{daloop}}
}
T = Load(\&Tail)\tcp*{Check if}
\If (\tcp*[f]{the queue is empty}) {\upshape T $\le$ H + 1} {
	catchup(T, H + 1)\;
	FAA(\&Threshold, -1)\;
	\Return $\varnothing$\;
}
\If {\upshape FAA(\&Threshold, -1) $\le$ 0} {\Return $\varnothing$}
}
}
\end{algorithm2e}
\end{subfigure}%
\caption{Scalable circular queue (SCQ).}
\label{alg:ringadv}
\end{figure}

SCQ leverages the high-level idea from the infinite array queue described above. However,
SCQ replaces SWAP operations with CAS, as memory buffers
are finite, and the same slot can be referenced by multiple cycles.

When discussing the threshold value, we previously assumed
that no enqueuer can fail when all dequeuers are behind. In SCQ, this
is no longer true, as the same entry can be occupied by some previous cycle.
To bound the maximum distance between the last dequeuer and
the last enqueuer, we \textit{\textbf{double the capacity of the queue}}
while still keeping the original number of elements. When using
this approach, all the enqueuers after the last dequeuer can always
locate an unused ($\bot$) slot no farther than $2n$ slots away from it.
The threshold value should now become $(n-1+2n)=3n-1$.

In the algorithm, we need a special value for $\bot$. We reserve the very last
index, $2n-1$, for this purpose. Since, in SCQ, $n$ is a power of 2 number,
$\bot$ will have all its index
bits set to~$1$. As we show below, this allows \textit{dequeue} to consume entries using an atomic OR operation.
This value does not overlap with the actual data indices, which are still less
than~$n$.

SCQ also accounts for additional corner cases that are not
present in the infinite queue:

\textbf{A dequeuer arrives prior to its enqueuer counterpart, but the
corresponding entry is already occupied.} This happens
when the entry is occupied by some other cycle.
If this cycle is already newer (Line~36 is false), \textit{dequeue} simply
retries because the enqueuer counterpart is guaranteed to fail (Line~16).
However, if the cycle is older, \textit{dequeue} needs to mark it accordingly,
so that when the enqueuer counterpart arrives, it will fail. For this purpose,
we clear the \textit{IsSafe} bit, as in CRQ. The key idea
is that the enqueuer will have to additionally make sure that all active
dequeuers are behind when \textit{IsSafe} is set to~0 (Line~16) before
inserting a new entry.
Whenever \textit{IsSafe} becomes~$0$, only enqueuers can change it back to~$1$.
(Only Line~17 sets the bit to~$1$; Line~35 preserves bit's value, and Line~33
sets it to~$0$.)

\textbf{Enqueuers attempt to use slots that are marked unsafe.}
If \textit{IsSafe} on Line~16 is~$0$, an enqueuer will additionally make sure
that the dequeuer that needs to be accounted for has not yet started by reading
and comparing the current \texttt{Head} value.

When dequeuing elements, if cycles match (Line~30), a dequeuer
is guaranteed to succeed. The corresponding slot will not be recycled
until an entry is consumed. The only thing that can change
is the \textit{IsSafe} bit.
Unlike CRQ, to mark an entry as consumed, \textit{dequeue} issues an atomic OR
operation which sets all index bits to~$1$
while preserving entry's safe bit and cycle.

The \textit{catchup} procedure is similar to the
\textit{fixState} procedure from CRQ and is used when the tail is behind
the head. This allows to avoid
unnecessary iterations in \textit{enqueue} and reduces the risk of
contention.

Finally, when comparing cycles, we use a common approach with signed integer
arithmetic which takes care of potential wraparounds.

\paragraph*{Optimization}
Similarly to LCRQ, SCQ employs an additional optimization on dequeuers.
If a dequeuer arrives prior to the corresponding enqueuer, it will
not aggressively invalidate a slot. Instead, it will spin for a small
number of iterations with the expectation that the enqueuer arrives soon. This alleviates unnecessary contention on the head and tail
pointers, and consequently helps both dequeuers and enqueuers.

\input{lring.tex}
\input{dcas.tex}

%% file: lring.tex
\subsection{SCQ-based unbounded queue (LSCQ)}
\label{sec:unbounded}

We follow LCRQ's main idea of maintaining a list of ring
buffers in our LSCQ design.
LSCQ is potentially more
memory efficient than LCRQ, as it is based on livelock-free SCQs
which do not end up being prematurely ``closed'' (Section~\ref{sec:eval}).
Since operations on the list are very rare, the cost
is completely dominated by SCQ operations.

\begin{figure}
\begin{subfigure}{.5\textwidth}
\begin{algorithm2e}[H]
\textbf{void *} ListHead = <empty \textbf{SCQ}>\;
\textbf{void *} ListTail = ListHead\;
\BlankLine
\Fn{\textbf{void} finalize\_SCQ(\textbf{SCQ} * cq)} {
Atomic\_OR(\&cq.Tail, \{.Value=0, .Finalize=1\})\;
}
\Fn{\textbf{void *} dequeue\_unbounded()} {
\While {\upshape True} {
  \textbf{SCQ *} cq = Load(\&ListHead)\;
  \textbf{void *} p = cq.dequeue\_ptr()\;
  \lIf {\upshape p $\ne$ \Null} {\Return p}
  \lIf {\upshape cq.next = \Null} {\Return \Null}
  Store(\&cq.aq.Threshold, $3n - 1$)\;
  p = cq.dequeue\_ptr()\;
  \lIf {\upshape p $\ne$ \Null} {\Return p}
  \If {\upshape CAS(\&ListHead, cq, cq.next)} {
    free\_SCQ(cq)\tcp*{Dispose of cq}
  }
}
}
\end{algorithm2e}
\end{subfigure}%
\begin{subfigure}{.5\textwidth}
\begin{algorithm2e}[H]
\setcounter{AlgoLine}{15}
\Fn{\textbf{void} enqueue\_unbounded(\textbf{void *} p)} {
  \While {\upshape True} {
  \textbf{SCQ *} cq = Load(\&ListTail)\;
  \If {\upshape cq.next $\ne$ \Null} {
    CAS(\&ListTail, cq, cq.next)\;
    \textbf{continue}\tcp*{Move list tail}
  }
  \tcp{Finalizes \& returns \textbf{false} if full}
  \If {\upshape cq.enqueue\_ptr(p, finalize=True)} {
    \Return\;
  }
  ncq = alloc\_SCQ()\tcp*{Allocate ncq}
  ncq.init\_SCQ(p)\tcp*{Initialize \& put p}
  \If {\upshape CAS(\&cq.next, \Null, ncq)} {
    CAS(\&ListTail, cq, ncq)\;
    \Return\;
  }
  free\_SCQ(ncq)\tcp*{Dispose of ncq}
  }
}
\end{algorithm2e}
\end{subfigure}%
\caption{Unbounded SCQ-based queue (LSCQ).}
\label{alg:lscq}
\end{figure}

In Figure~\ref{alg:lscq},
we present the LSCQ algorithm. We intentionally ignore the memory reclamation
problem (Section~\ref{sec:background}) which can be straight-forwardly solved
by the corresponding techniques such as hazard pointers~\cite{hazardPointers}.
The presented unbounded queue can store any fixed-size data entries, including
pointers (as in Figure~\ref{alg:lscq}), just like SCQ itself. When a ring
buffer is full,
we need to additionally ``finalize'' it, so that no new entries are
inserted by the concurrent threads. As in CRQ, we reserve one bit
in \texttt{Tail}. When \textbf{fq} does not have available entries (Line~3,
Figure~\ref{alg:ringptr}), we set the corresponding bit for
\textbf{aq}'s \texttt{Tail}.

We also modify \textit{enqueue} for \textbf{aq} such that it fails when
FAA on \texttt{Tail} returns a value with the ``finalized'' bit set.
Thus, any concurrent thread that tries to insert entries after \textbf{aq}
is being finalized, fails. In this case, we also need to place the entry back
into \textbf{fq}. This cannot fail since \textbf{fq} is never finalized.

Before unlinking \textit{cq} from the list (Line~14),
\textit{dequeue\_unbounded} checks again that \textit{cq}, which must already be finalized, is empty. Pending enqueuers may still access it. For the final
check, the threshold must be reset so that slots for the pending enqueuers
can be invalidated.

%% file: dcas.tex
\subsection{SCQ for double-width CAS}
\label{sec:dcasscq}

The x86-64~\cite{intel:manual} and ARM64~\cite{arm:manual} architectures implement double-width
CAS, which atomically updates two contiguous words. We can leverage this
capability to build SCQ which avoids indirection when storing arbitrary
pointers. In this case, all entries consist of tuples. Each tuple
stores two adjacent integers instead of just one integer.
The \textit{index} field
of the first integer now simply indicates
if the entry is occupied~($0$) or available~($\bot$). The second
integer from the same tuple stores a pointer
which is used in lieu of an index.
Since pointers also need to be stored and retrieved, we change Lines~18, 31, and 37
to use double-width CAS accordingly.

This version of SCQ provides a fully compatible API such that
architectures without double-width CAS can still
implement the same queue through indirection.
In this queue, \textit{enqueue} becomes \textit{enqueue\_ptr}, and
\textit{dequeue} becomes \textit{dequeue\_ptr}.

If \textit{enqueue\_ptr} needs to identify full queues,
additional changes are required. In Figure~\ref{alg:dcas}, we show
a method which compares \texttt{Head} and \texttt{Tail} values.
The comparison is
relaxed and up to $k$ ($k\le n$) concurrent enqueuers can increment
\texttt{Tail} spuriously. Thus, \texttt{Tail} can now be up to $3n$ slots
ahead of \texttt{Head}. Since we previously assumed that number
to be $2n$, we increase the threshold from $3n-1$ to
$4n-1$. This method is imprecise and can only guarantee that at least $n$
elements are in the queue, but the actual number varies.
This is often acceptable, especially when creating
unbounded queues (Section~\ref{sec:unbounded}), which finalize full queues.

\begin{figure}
\begin{algorithm2e}[H]
\Fn{\textbf{bool} enqueue\_ptr(\textbf{void *} ptr)} {
\tcp{Add a full queue check before Line~12:}
T = Load(\&Tail)\;
\lIf {\upshape T $\ge$ Load(\&Head) + 2n} {\Return False}
... Modified enqueue() ...

\tcp{Add a full queue check in the loop after Line~22:}
\lIf {\upshape T+1 $\ge$ Load(\&Head) + 2n} {\Return False}
}
\end{algorithm2e}
\caption{SCQ for double-width CAS: checking for full queues.}
\label{alg:dcas}
\end{figure}

%% file: correctness.tex
\section{Correctness}

We omit more formal linearizability arguments for NCQ due to its
simplicity.
SCQ's linearizability follows from
the arguments we make in Sections~\ref{sec:infin}~and~\ref{sec:scqalg}, as well as from
the corresponding CRQ linearizability
derivations~\cite{Morrison:2013:FCQ:2442516.2442527} because the SCQ
design has many similarities with CRQ.
Below we provide lock-freedom arguments for NCQ and SCQ.

\begin{theorem}
The NCQ algorithm is lock-free.
\end{theorem}

\begin{proof}
The NCQ algorithm has two unbounded loops: one in \textit{enqueue} (Lines~5-15) and the other one in \textit{dequeue} (Lines~18-26).

If CAS fails in \textit{enqueue} causing it to repeat,
it means that the corresponding entry \texttt{Entries[j]} is changed
by another thread executing \textit{enqueue}, as \textit{dequeue}
does not modify entries. Consequently, that other thread is making
progress, i.e., succeeding in the \textit{enqueue} operation (Line~15).

If CAS fails in \textit{dequeue} causing it to repeat,
it means that the \texttt{Head} pointer is modified
by another thread executing \textit{dequeue}. Therefore, that other
thread is making progress, i.e., succeeding in the \textit{dequeue} operation (Line~26).
\end{proof}

\begin{theorem}
The SCQ algorithm is lock-free.
\end{theorem}

\begin{proof}
The SCQ algorithm has two unbounded loops: one in \textit{enqueue} (Lines~12-22) and the other one in \textit{dequeue} (Lines~26-45).

If the condition on Line~16 of \textit{enqueue} is false causing
it to repeat the loop, then two possibilities exist.
First, some dequeuer already invalidated the slot for this enqueuer
(\textit{Cycle(Ent)} and \textit{IsSafe(Ent)} checks) because it
arrived before the enqueuer.
Alternatively, the entry is occupied by a prior enqueuer (i.e., $\ne \bot$)
and is supposed to be consumed by some dequeuer which is yet to come.

In the latter case, the current enqueuer skips the occupied slot.
There may also exist other concurrent enqueuers which will skip
occupied slots as well.
Since the total number of elements for \textit{enqueue} is always capped,
the queue will never have more than $n$ entries. Thus, the enqueuers should
be able to succeed unless dequeuers keep invalidating their slots.

The first case is more intricate. If none of the enqueuers succeed,
dequeuers must not be able to
invalidate new slots after a finite number of steps. Once dequeuers stop
invalidating slots (using the threshold described in
Section~\ref{sec:infin}),
at least one enqueuer can make further progress by catching up
its \texttt{Tail} to the next available position and inserting a new element.

If the conditions on Lines~30, 40, or 44 of \textit{dequeue} are false
causing it to repeat the loop, then the following must be the reason for that.
Line~30 can only be false if the entry is not yet initialized
by the corresponding enqueuer. In this case, the dequeuer may potentially
iterate and invalidate slots as many as $3n$ times until either Line~43 or 45
terminates the loop. Line~45 is guaranteed to eventually terminate the loop as
long as no new entries are inserted by \textit{enqueue} (i.e., enqueuers
can be running, but none of them succeed). Any pending (almost completed)
enqueuer may still cause Line~21 to increase the threshold value temporarily
even though its entry was already consumed.
However, this at most is going to happen for $k-1$ pending
enqueuers. After that, the threshold value will be depleted causing
all active dequeuers to complete (Line~45). Since dequeuers are no longer
running, at least one new
enqueuer will be able to succeed. At that point, it will reset
the threshold value (Line~21). After that, we recursively
repeat this entire argument again to show that at least one
following enqueuer will succeed in a finite number of steps.
\end{proof}

%% file: evaluation.tex
\section{Evaluation}
\label{sec:eval}

In this section, we evaluate our SCQ design against
well-known or state-of-the-art algorithms. We use and extend the
benchmark from~\cite{Yang:2016:WQF:2851141.2851168} which already
implements several algorithms.

In the evaluation, we show that SCQ achieves very high performance
while avoiding limitations that are typical to other high-performant
algorithms (i.e., livelock workarounds, memory reclamation, and
portability). We have also found
that state-of-the-art approaches, especially LCRQ, can have very high memory usage, a problem that does not exist in SCQ.

We present results for both bare-bones SCQ and the version that stores
arbitrary pointers (SCQP). Bare-bones SCQ is relevant because queue
elements in SCQ can be of any type, i.e., not necessarily pointers.
We compare SCQ and SCQP against M\&S FIFO
lock-free queue (MSQUEUE)~\cite{Michael:1998:NAP:292022.292026}, combining
queue (CCQUEUE)~\cite{Fatourou:2012:RCS:2145816.2145849} -- which is not a lock-free
queue but is known to have good performance, LCRQ~\cite{Morrison:2013:FCQ:2442516.2442527} -- a queue
that maintains a lock-free list of ring buffers (CRQ); CRQs
cannot be used separately, as they are susceptible to livelocks,
WFQUEUE -- a recent scalable wait-free queue design~\cite{Yang:2016:WQF:2851141.2851168}. These algorithms provide reasonable baselines as they represent
well-known or state-of-the-art (in scalability) approaches.
We also present NCQ as an additional baseline. NCQ uses the same data structure
as SCQ, but its design is reminiscent of MSQUEUE.
Finally, we provide FAA (fetch-and-add) throughputs to
show the potential for scalability. FAA is not a real algorithm,
it simply implements atomic increments on \texttt{Head} and \texttt{Tail} when
calling \textit{dequeue} and \textit{enqueue} respectively.
We skip a separate LSCQ evaluation since SCQ is already
lock-free and can be used as is. LSCQ's costs are largely dominated by the
underlying SCQ implementation.

Performance differences between queues should be treated with great
caution. For example, LCRQ has better throughput
in a few tests, but it consumes a lot of memory under
certain circumstances. Also, neither LCRQ nor WFQUEUE are deployable
in all places where SCQ
can be used, e.g., fixed-size data pools. As mentioned in introduction,
this would trigger a ``chicken and egg'' situation, where data pools would
need to depend on some external (typically non lock-free) memory allocator.
Moreover, both LCRQ and WFQUEUE require safe memory reclamation by their
design;
the benchmark implements a specialized scheme for WFQUEUE and hazard
pointers~\cite{hazardPointers} for LCRQ and MSQUEUE. Finally, WFQUEUE
needs special per-thread descriptors, which reduce the API transparency.

We run all experiments for up to 144 threads on a 72-core machine consisting
of four Intel Xeon E7-8880~v3 CPUs with hyper-threading disabled, each running at
2.30~GHz and with a 45MB~L3 cache. The machine has 128GB of RAM and runs
Ubuntu 16.04~LTS. We use gcc 8.3 with the -O3 optimization flag. For SCQP, we use the double-width version (Section~\ref{sec:dcasscq}) as x86-64 can benefit from it. SCQP is compiled with
clang~7.0.1 (-O3) because it generates faster code for our implementation
(no such advantage for other queues).

We also evaluate queues on the PowerPC architecture. We run
experiments on an 8-core POWER8
machine. Each core has 8 threads, so we have 64 logical cores in total.
We do not disable PowerPC's simultaneous multithreading because the machine
does not have as many cores as our Xeon testbed.
All cores are running at 3.0 Ghz and have an 8MB~L3 cache. The machine has
64GB of RAM and runs Ubuntu 16.04~LTS. We use gcc 8.3 with
the -O3 optimization flag.
Since PowerPC does not support double-width CAS, it is impossible
to implement the LCRQ algorithm there. In contrast, SCQ and SCQP
both work well on PowerPC.
For SCQP, we use the version with indirection and two queues.

\begin{figure}[ht]
\begin{subfigure}{.5\textwidth}
\includegraphics[width=\textwidth]{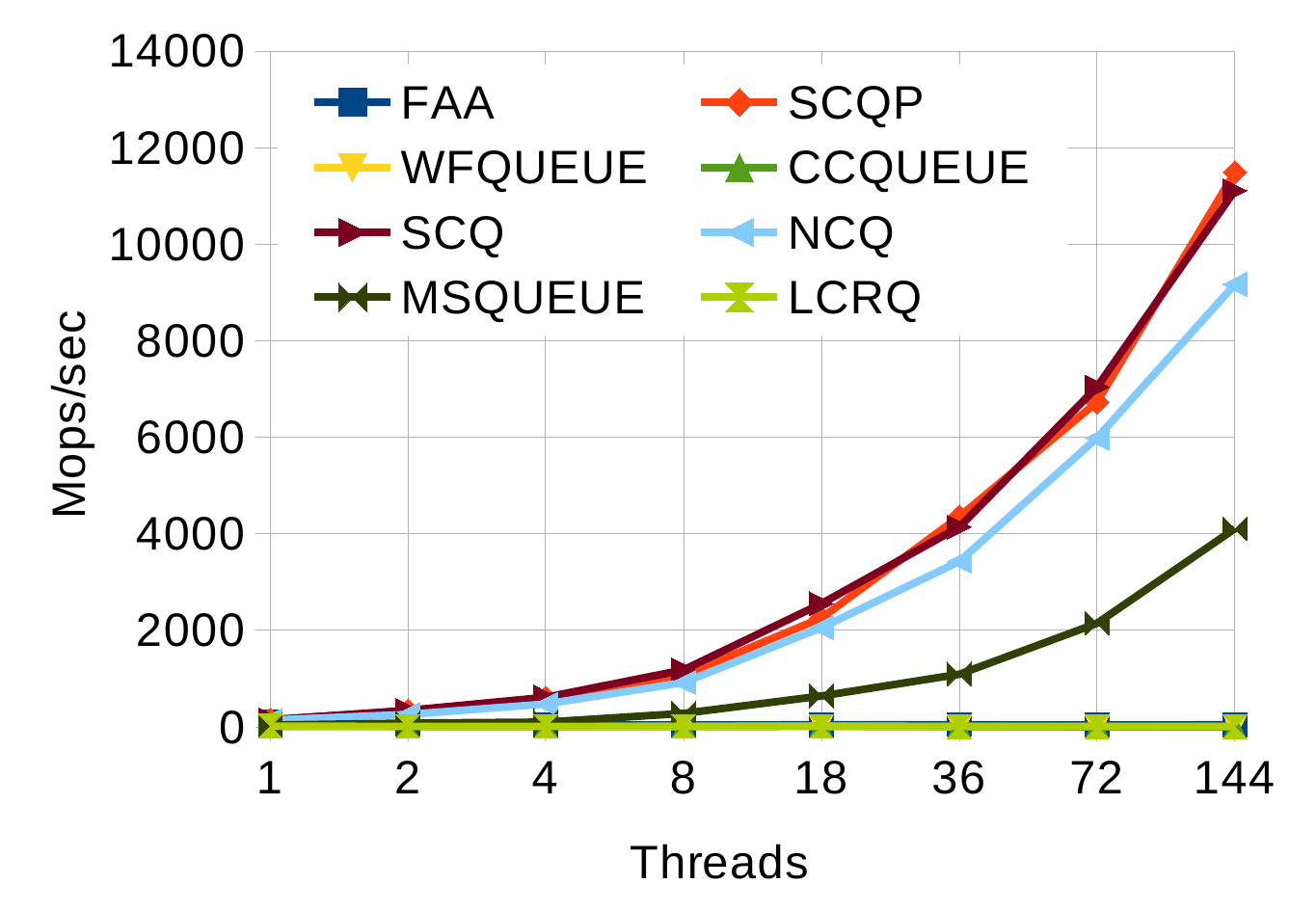}
\caption{4x18-core Intel Xeon E7-8880}
\label{fig:emptyx86}
\end{subfigure}%
\begin{subfigure}{.5\textwidth}
\includegraphics[width=\textwidth]{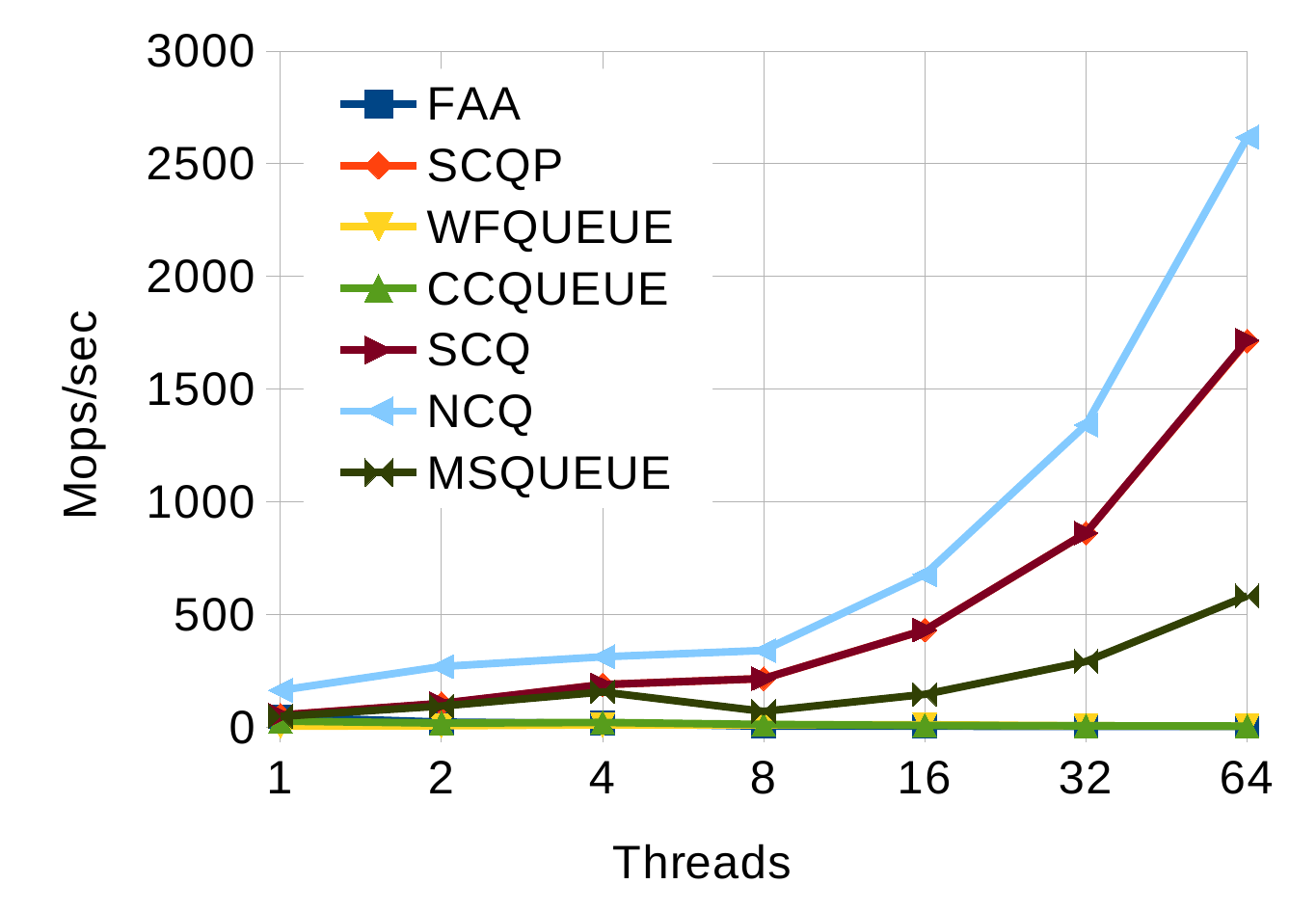}
\caption{8x8-core POWER8}
\label{fig:emptyppc}
\end{subfigure}%
\caption{Empty queue test, throughput of the dequeue operation.}
\label{fig:empty}
\end{figure}

We use jemalloc~\cite{jemalloc} to alleviate libc's malloc poor
performance~\cite{allocators}. Each data point is measured 10 times for
10000000 operations in a loop, we present the average.
The benchmark measures
throughput in a steady/hot state and protects against occasional outliers.
We use the default benchmark parameters from~\cite{Yang:2016:WQF:2851141.2851168}
to achieve optimal performance with LCRQ, CCQUEUE, and WFQUEUE.
For SCQ and NCQ, we have chosen a relatively small ring buffer size,
$2^{16}$ entries. (SCQ uses only half queue's capacity, $n=2^{15}$
entries, as discussed in Section~\ref{alg:ringadv}.) LCRQ uses $2^{12}$ entries in
each CRQ to
attain optimal
performance. Unlike SCQ, LCRQ wastes a lot of memory in each CRQ due to
cache-line padding.
Most of our results for x86-64 have peaks for 18 threads because each
CPU has 18 cores. Over-socket contention is expensive
and results in performance drops. PowerPC has an analogous picture.

In Figure~\ref{fig:empty}, we perform a simple experiment to measure
the cost of \textit{dequeue} on empty queues. MSQUEUE, NCQ, SCQ, and SCQP
perform reasonably well on both x86-64
(Figure~\ref{fig:emptyx86}) and PowerPC (Figure~\ref{fig:emptyppc})
since they do not dequeue elements aggressively.
LCRQ, WFQUEUE, and CCQUEUE take a performance hit in this
corner case. FAA is also slower than MSQUEUE, NCQ, SCQ, and SCQP because
it still needs to modify an atomic variable.
Slow dequeuing on empty queues was previously acknowledged
by WFQUEUE's authors~\cite{Yang:2016:WQF:2851141.2851168}.

\begin{figure}[ht]
\begin{subfigure}{.5\columnwidth}
\includegraphics[width=\textwidth]{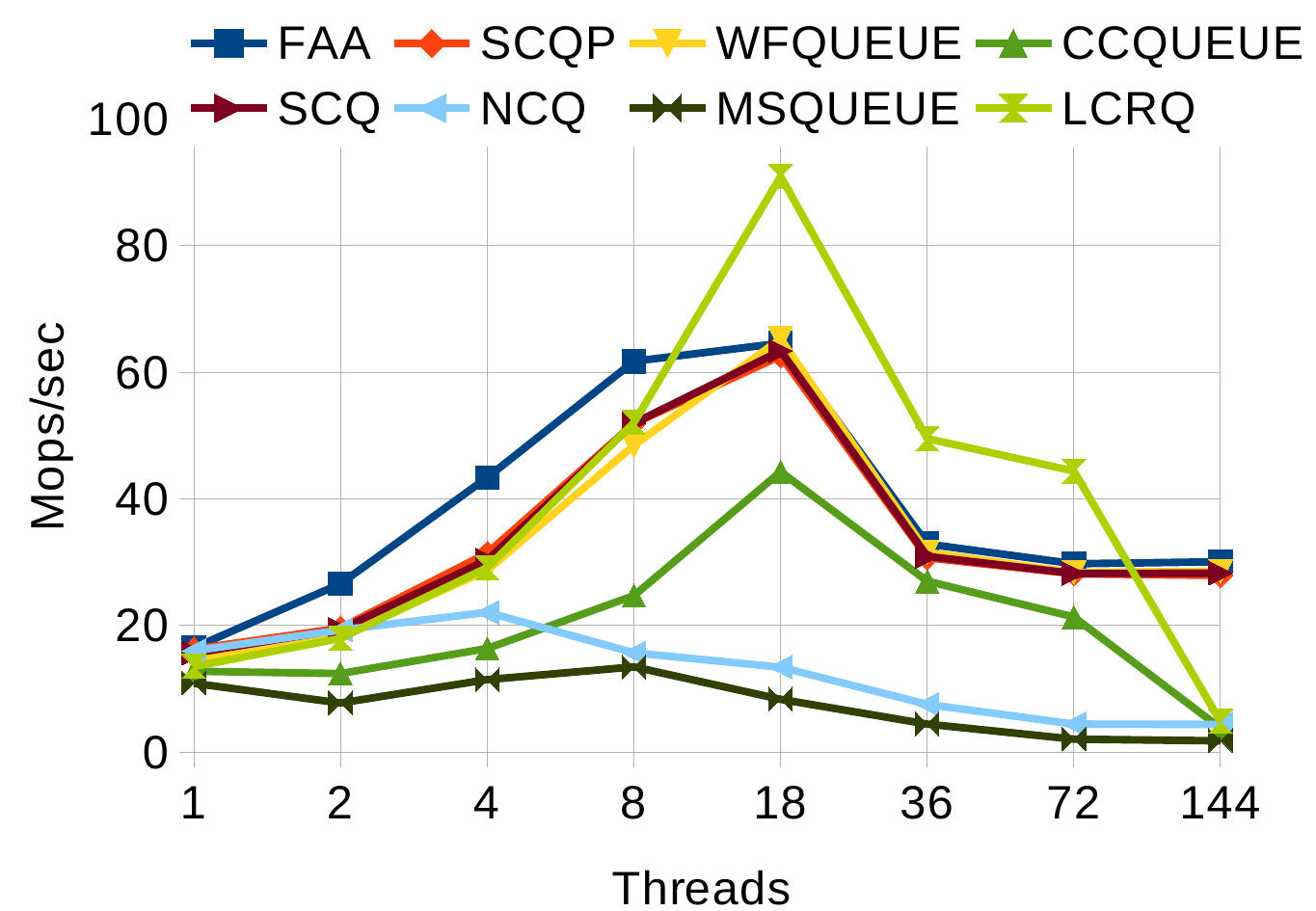}
\caption{Throughput (higher is better)}
\label{fig:halfhalfdel}
\end{subfigure}%
\begin{subfigure}{.5\columnwidth}
\includegraphics[width=\textwidth]{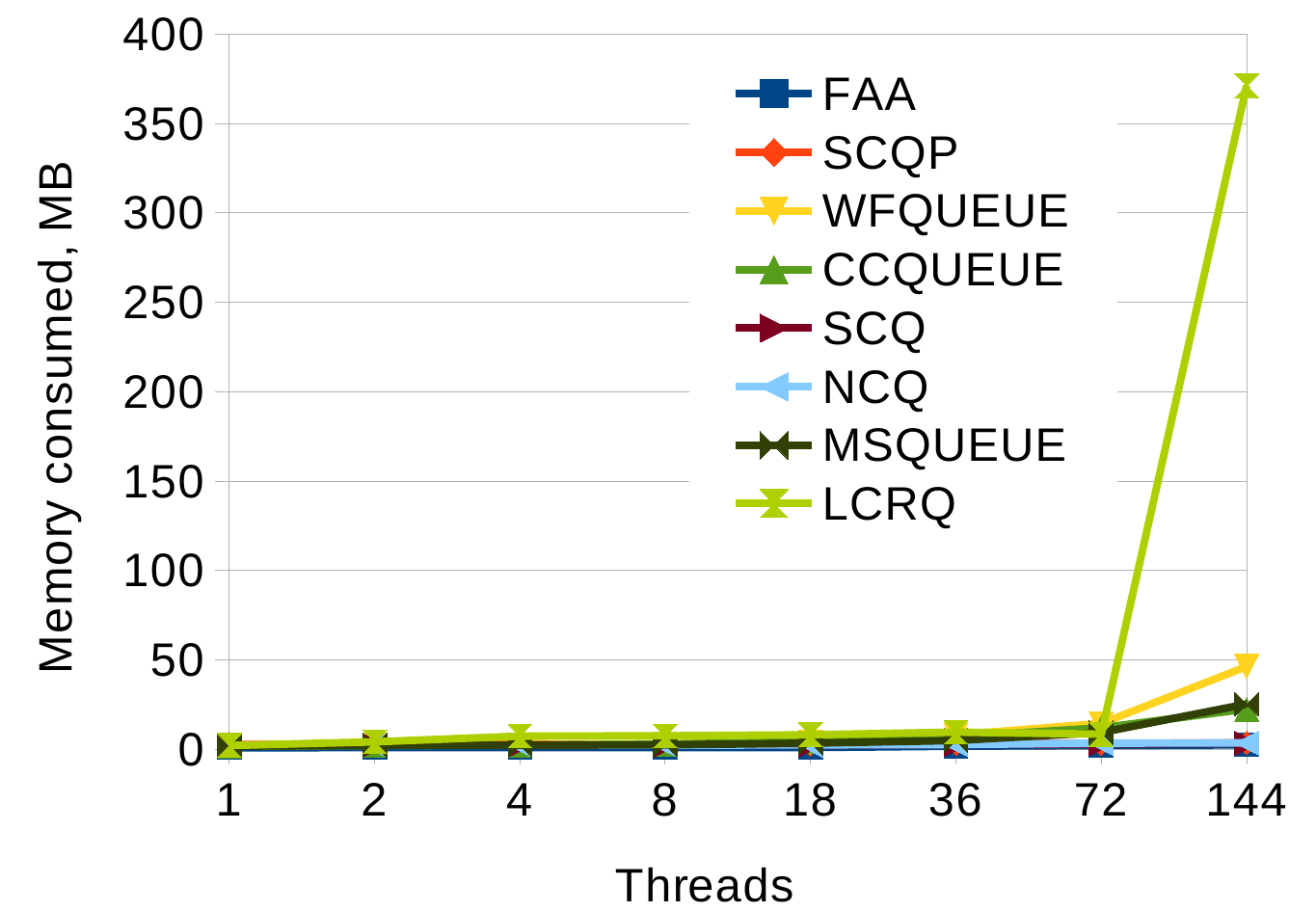}
\caption{Memory consumption (lower is better)}
\label{fig:halfhalfmem}
\end{subfigure}%
\caption{Memory efficiency test, 4x18-core Intel Xeon E7-8880 (standard malloc).}
 \label{fig:memory}
\end{figure}

To evaluate memory efficiency, we run an experiment with
50\% of \textit{enqueue} and 50\% of \textit{dequeue} operations that are
chosen by each thread randomly
(Figure~\ref{fig:memory}). For this test, we use libc's standard malloc
to make sure that memory pages are unmapped more aggressively.
We run the benchmark with its default configuration
that uses tiny delays between operations.
Using delays allows us to get more pronounced results while still showing
a realistic execution scenario.
It turns out that while, for the most part, LCRQ provides
higher throughput (Figure~\ref{fig:halfhalfdel}), it can also allocate
a lot of memory while running (Figure~\ref{fig:halfhalfmem}), up to
$\approx400$MB.
WFQUEUE's memory usage is also somewhat elevated (up to $\approx50$MB)
and exceeds that of MSQUEUE and CCQUEUE for most data points.
Conversely, SCQ, SCQP, and NCQ are very efficient; they only need a
small (512K-1MB), fixed-size buffer that is allocated for circular
queues. Overall, SCQ and SCQP win here as they both achieve great performance
with very little memory overhead.

Since the design of SCQ is related to LCRQ, we were particularly interested
in investigating LCRQ's high memory usage. As we suspected,
LCRQ was ``closing'' CRQs frequently due to livelocks. To maintain
good performance, LCRQ must use relatively large CRQs
(each of them has $2^{12}$ entries). However, due to livelocks, CRQs need to be
prematurely closed from time to time. Eventually, LCRQ ends up in a situation
where it
frequently allocates new CRQs, i.e., wasting memory greatly. Safe memory
reclamation additionally impacts the timing of deallocation, i.e., CRQs are not
deallocated immediately.

\begin{figure}[ht]
\begin{subfigure}{.5\columnwidth}
\includegraphics[width=\textwidth]{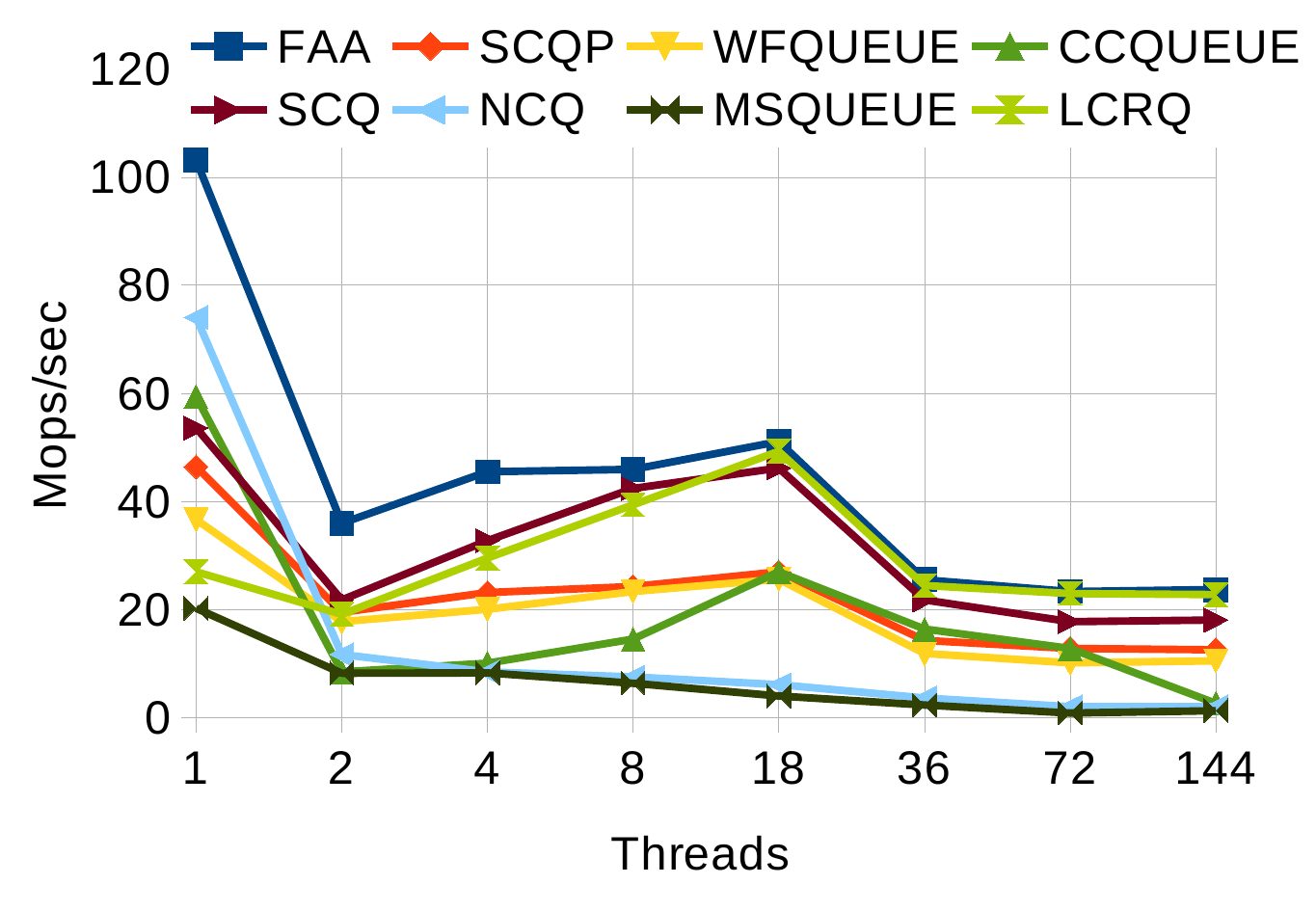}
\caption{enqueue-dequeue pairs}
\label{fig:bpairwise}
\end{subfigure}%
\begin{subfigure}{.5\columnwidth}
\includegraphics[width=\textwidth]{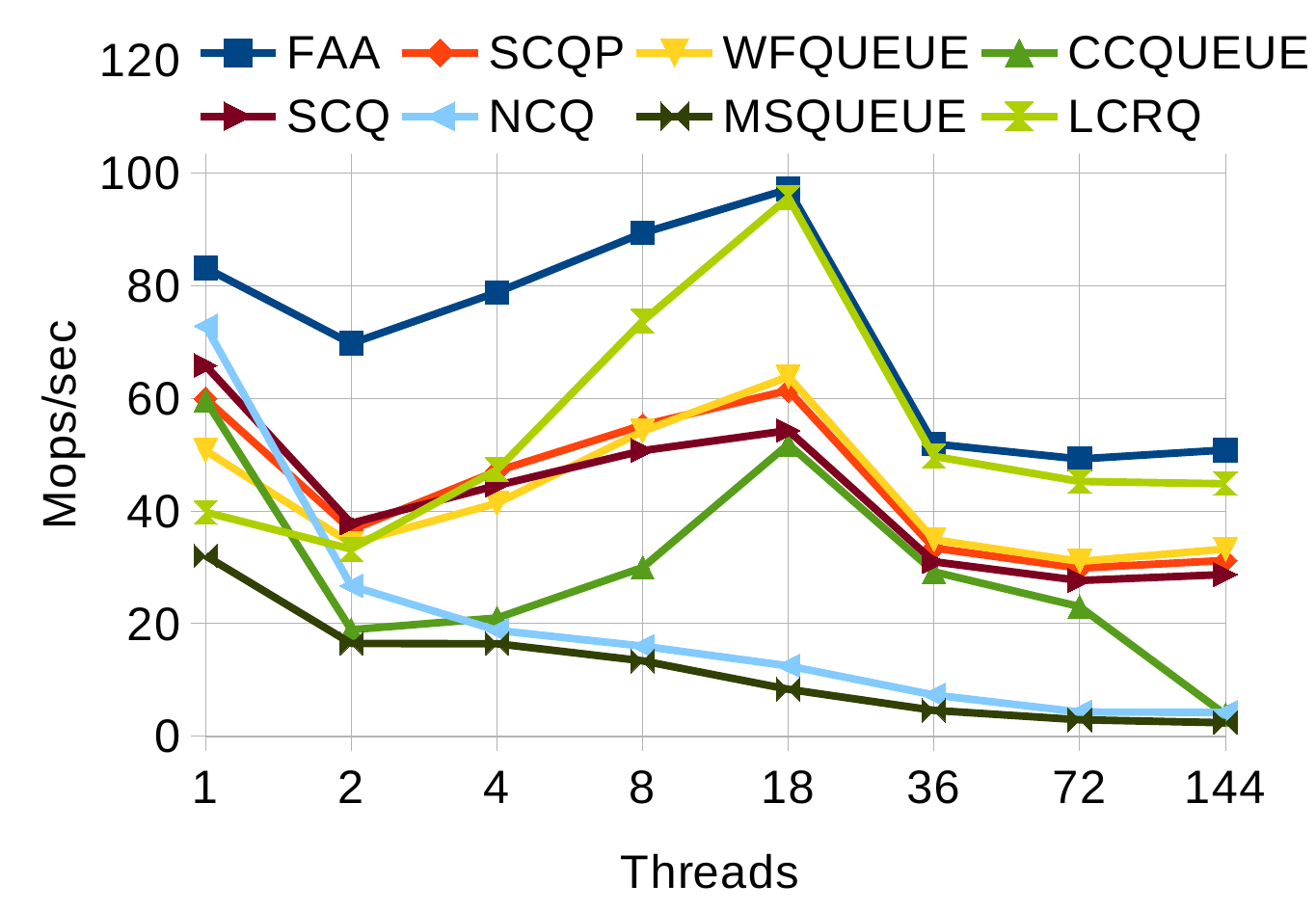}
\caption{50\% enqueues, 50\% dequeues}
\label{fig:bhalfhalf}
\end{subfigure}%
\caption{Balanced load tests, 4x18-core Intel Xeon E7-8880.}
\label{fig:uuux}
\end{figure}

Finally, we evaluate queues using operations by
multiple threads in a tight loop. In Figures~\ref{fig:bpairwise} and~\ref{fig:apairwiseppc}, we present throughput for the x86-64 and PowerPC architectures
respectively when using pairwise queue operations. In this experiment,
every thread executes \textit{enqueue} followed by \textit{dequeue} in a tight
loop. Since multiple concurrent threads
are running simultaneously, the order of dequeued elements is not
predetermined anyhow (even though enqueue and dequeue are paired).
For x86-64, SCQ and LCRQ are both winners and have roughly similar
performance. SCQP, WFQUEUE, and CCQUEUE (partially) attain half of their
throughput
on average. For PowerPC, SCQ is a winner: it generally outperforms
all other algorithms. SCQP and WFQUEUE are roughly identical, except
that WFQUEUE marginally outperforms SCQP for smaller concurrencies.
CCQUEUE is generally worse, except very small concurrencies.

\begin{figure}[ht]
\begin{subfigure}{.5\columnwidth}
\includegraphics[width=\textwidth]{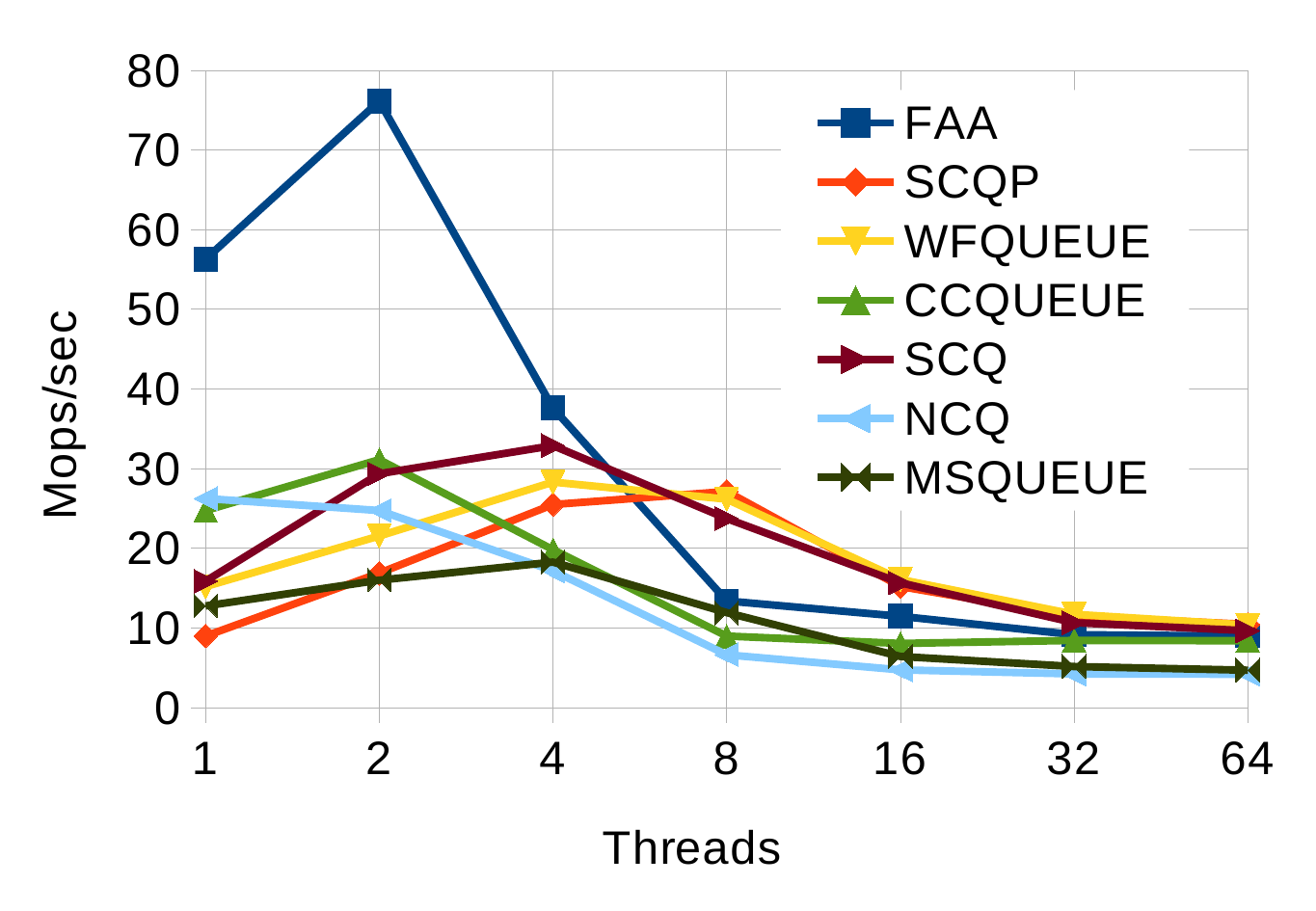}
\caption{enqueue-dequeue pairs}
\label{fig:apairwiseppc}
\end{subfigure}%
\begin{subfigure}{.5\columnwidth}
\includegraphics[width=\textwidth]{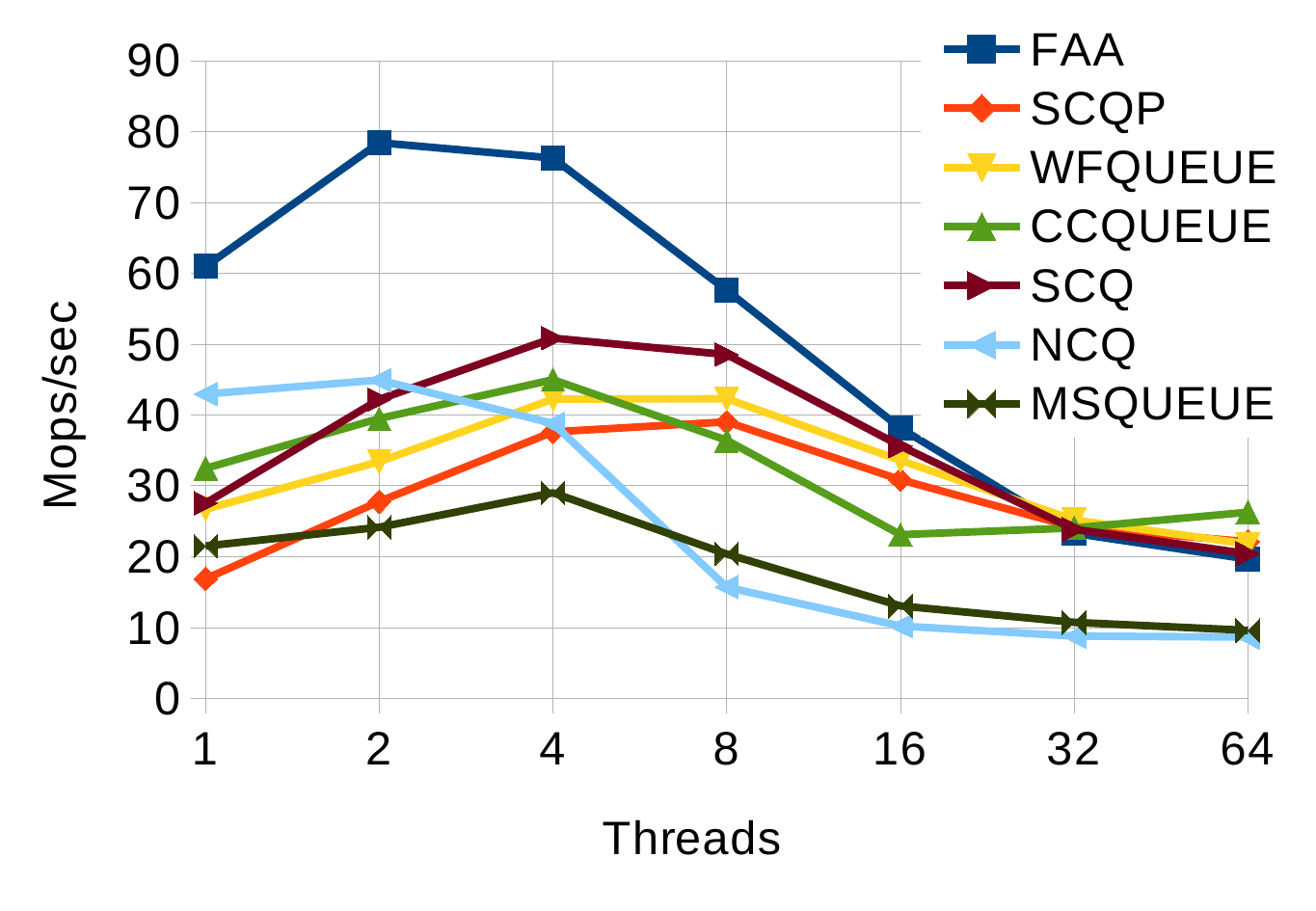}
\caption{50\% enqueues, 50\% dequeues}
\label{fig:ahalfhalfppc}
\end{subfigure}%
\caption{Balanced load tests, 8x8-core POWER8.}
\label{fig:ppc}
\end{figure}

In Figures~\ref{fig:bhalfhalf} and~\ref{fig:ahalfhalfppc}, we present results
for an experiment which selects operations randomly: 50\% of enqueues and 50\% of dequeues.
For x86-64, WFQUEUE and SCQP are almost identical. SCQP marginally
outperforms SCQ when concurrency is high, probably due to (occasional) cache
contention since entries are larger in double-width SCQP.
CCQUEUE is typically slower than WFQUEUE, SCQP, or SCQ.
LCRQ outperforms all of them most of the time. However, considering its
memory utilization, LCRQ may not be appropriate in a number of cases as
previously discussed.
For PowerPC, SCQ is a winner: it generally outperforms all other
algorithms. SCQP, WFQUEUE, and CCQUEUE are very close; WFQUEUE marginally
outperforms SCQP in this test.

%% file: related.tex
\section{Related Work}

Over the last couple of decades, different designs were proposed for
concurrent FIFO queues as well as ring buffers.
A classical Michael \& Scott's lock-free FIFO queue~\cite{Michael:1998:NAP:292022.292026} maintains
a list of nodes. The list has the \textit{head} and \textit{tail} pointers.
These pointers must be updated using CAS operations. The queue can be
modified to avoid the ABA problem related to pointer updates. For that purpose,
double-width CAS is used to store an ABA tag for each pointer.

Existing lock-free ring buffer designs that do not benefit from FAA (e.g.,~\cite{Tsigas:2001:SFS:378580.378611}) are typically not very scalable.
Another approach~\cite{Feldman:2015:WMM:2835260.2835264}, though uses FAA,
needs a memory reclamation scheme and does not seem to scale as well
as some other algorithms.
Certain queues use FAA but are not linearizable. 
For example, \cite{Freudenthal:1991:PCF:106972.106998} maintains a queue size
and updates it with FAA. However, the queue may end up in inconsistent
state as previously discussed in~\cite{Morrison:2013:FCQ:2442516.2442527}. 

Certain bounded MPMC queues without explicit locks such
as~\cite{Krizhanovsky:2013:LMM:2492102.2492106,vyakov}
are relatively fast. However, these approaches are technically not
lock-free as discussed in~\cite{Feldman:2015:WMM:2835260.2835264,ringdisapp}.
Just as with explicit spin locks, lack of true lock-freedom manifests
in suboptimal performance when threads are preempted,
as remaining threads are effectively blocked when the preemption happens
in the middle of a queue operation.

Alternative concurrent designs were also considered.
CCQUEUE~\cite{Fatourou:2012:RCS:2145816.2145849} is a special combining
queue. The reasoning behind this queue is that it may be cheaper to execute
operations sequentially. It reuses the same node for enqueue-dequeue
pairs, so that no cache misses occur due to contention.

Finally, the use of FAA for queues is advocated by recent non-blocking FIFO
queue designs~\cite{Yang:2016:WQF:2851141.2851168,Morrison:2013:FCQ:2442516.2442527}. Unfortunately,~\cite{Yang:2016:WQF:2851141.2851168,Morrison:2013:FCQ:2442516.2442527}, despite their good performance,
are not always memory efficient.

%% file: conclusion.tex
\section{Conclusion}
In this paper, we presented SCQ, a scalable lock-free FIFO queue.
The main advantage of SCQ is that the queue is standalone,
memory efficient, ABA safe, and scalable. At the same time, SCQ does not
require a safe memory reclamation scheme or memory allocator.
In fact, this queue itself can be leveraged for object allocation
and reclamation, as in data pools.

We use FAA (fetch-and-add) for the most contended hot spots of the algorithm:
\texttt{Head} and \texttt{Tail} pointers. Unlike prior
attempts to build scalable queues with FAA such as CRQ, our queue is both
lock-free and linearizable. SCQ prevents livelocks directly in the ring
buffer itself without trying to work around the problem by allocating
additional ring buffers and linking them together.
SCQ is very portable and can be implemented virtually everywhere. It only
needs single-width CAS.
It is also possible to create unbounded queues based
on SCQs which are more memory efficient than LCRQ.

We thank the anonymous reviewers for their valuable feedback.

SCQ's source code is available at \url{https://github.com/rusnikola/lfqueue}.